\documentclass[letterpaper, 10 pt, conference]{ieeeconf}  %

\IEEEoverridecommandlockouts                              %

\usepackage[utf8]{inputenc}

\usepackage{graphics} %
\usepackage{epsfig} %
\usepackage{amsmath} %

\usepackage{bm}
\usepackage{graphicx}
\usepackage{xcolor}
\usepackage[english]{babel}
\usepackage{amsfonts}
\usepackage{amssymb}
\usepackage{amsbsy}
\usepackage{hyperref}
\usepackage{epstopdf}
\usepackage{caption}
\usepackage{subcaption}

\usepackage{tabularx}
\usepackage{dsfont}
\usepackage{amsopn}
\usepackage{mathtools}
\usepackage{enumerate}  
\usepackage{balance}

\usepackage{amsthm}

\usepackage{irt-common}
\usepackage{irt-shortcuts}

\usepackage{cleveref}

\newtheoremstyle{myplain}
{\topsep}   %
{\topsep}   %
{\itshape}  %
{0pt}       %
{\bfseries} %
{}         %
{5pt plus 1pt minus 1pt} %
{}          %

\theoremstyle{myplain}

\newtheorem{theorem}{Theorem}
\newtheorem{proposition}[theorem]{Proposition}
\newtheorem{lemma}[theorem]{Lemma}

\newtheoremstyle{mydefinition}
{\topsep}   %
{\topsep}   %
{}  %
{0pt}       %
{\bfseries} %
{}         %
{5pt plus 1pt minus 1pt} %
{}          %

\theoremstyle{mydefinition}
\newtheorem{definition}[theorem]{Definition}
\newtheorem*{definition*}{Definition}

\newcommand{\change}[1]{#1}

\renewcommand{\vec}[1]{\ensuremath{\mathbf{#1}}}

\title{\LARGE \bf
Uniform Detectability of Linear Time Varying Systems with Exponential Dichotomy
}

\author{
	Markus Tranninger$^1$, Richard Seeber$^2$, Martin Steinberger$^1$, and Martin Horn$^{1,2}$
	\thanks{$^1$ Institute of Automation and Control, Graz University of Technology, 8010 Graz, Austria.
		{\tt\small markus.tranninger@tugraz.at}}
	\thanks{$^2$ Christian Doppler Laboratory for Model Based Control of Complex Test Bed Systems, Institute of Automation and Control, Graz University of Technology, Graz, Austria.
		{\tt\small richard.seeber@tugraz.at}}
	\thanks{This work was partially supported by the Graz University of Technology LEAD project ``Dependable Internet of Things in Adverse Environments''. The financial support by the Christian Doppler Research Association, the Austrian Federal Ministry for Digital and Economic Affairs and the National Foundation for Research, Technology and Development is gratefully acknowledged.}
}

\begin{document}

\maketitle
\thispagestyle{empty}
\pagestyle{empty}

\begin{abstract}
Exponential dichotomies play a central role in stability theory for dynamical systems. 
They allow to split the state space into two subspaces, where all trajectories in one subspace decay whereas all trajectories in the other subspace grow, uniformly and exponentially.
This paper studies uniform detectability and observability notions for linear time varying systems, which admit an exponential dichotomy. 
The main contributions are necessary and sufficient detectability conditions for this class of systems.
\end{abstract}

\section{Introduction}

Detectability and observability are essential concepts in the observer design for dynamical systems. 
In the linear time invariant case, it is well known that detectability is equivalent to asymptotic stability of the unobservable part of the system. 
For time varying systems, it is harder to obtain such insights, and attempts to extend this characterization were recently made in~\cite{Frank2018,Bocquet2017,Tranninger2020}.
The main difficulty is that detectability is inherently linked to the uniformity of stability properties of the underlying observer error dynamics, see, e.g.~\cite{Tranninger2020,Kalman1961}. 
In the linear time invariant case, exponential and uniform exponential stability notions coincide. They differ significantly and show different robustness properties, however, in the time varying case~\cite{Anderson2013}. 

For periodic, analytic and so-called constant rank systems, a smooth Kalman decomposition exists and allows to split the state space into an observable and unobservable part~\cite{Tai1987,Kalman1962}. 
Hence, similar insights as in the time invariant case are possible in these cases, but the system classes are quite restrictive. 
In the general linear time varying case, the desired system decomposition is hard to achieve. 
Moreover, it might be not continuous, i.e., the dimensions of the resulting subsystems change over time. 
Detectability conditions for linear time varying systems are rarely studied in literature. 
In~\cite[Lemma 3.4]{Ravi1992}, detectability is linked to the existence of uniformly bounded solutions of the filter Riccati equation, which guarantees the existence of an observer~\cite{Kalman1961}.
This does not allow to obtain the same insights as in the time invariant case, however.

Exponential dichotomies allow a different splitting of the state space into two invariant subspaces. 
All solutions in one subspace decay exponentially and uniformly with respect to the initial time (stable subspace), whereas the solutions in the other subspace grow uniformly with an exponential rate (unstable subspace). 

Detectability or the dual concept stabilizability for systems with exponential dichotomy are rarely studied in literature. 
The only work the authors are aware of is~\cite{Ilchmann1987}. 
There, the concept of controllability into subspaces is introduced and it is shown that controllability into the stable subspace is necessary and sufficient for stabilizability.
Recently, numerical tools, which allow to determine if a system has an exponential dichotomy and compute the splitting numerically were presented in~\cite{Dieci2010a,Dieci2011}. 
This motivates further investigations on detectability for systems with exponential dichotomy in order to obtain computable detectability conditions.
The present work is a step towards this goal.

\change{
The main result is a necessary and sufficient detectability condition for systems with exponential dichotomy in block triangular form.
This system class is relevant, because the continuous QR-decomposition allows a well-conditioned transformation of a given linear system to upper triangular form
and, moreover, the numerical tools in order to detect an exponential dichotomy are based on this upper triangular form~\cite{Dieci2011}.
In the course of proving the main result, it is furthermore shown that for anti-stable systems, uniform complete observability is necessary and sufficient for the considered detectability property.
}

The paper is structured as follows. 
\Cref{sec:problemstatement} presents the considered problem statement in detail.
The concept of an exponential dichotomy and the required observability notions are recalled in~\Cref{sec:EDbasics}. 
The main result, i.e., necessary and sufficient detectability conditions for systems with exponential dichotomy in block triangular form are presented in~\Cref{sec:Detectability}.
Moreover, the required tools to \change{prove} the main result are established. 
\Cref{sec:conclusion} discusses the achieved results and points out possible applications and future research directions.

\emph{Notation:} Matrices are printed in boldface capital letters, whereas column vectors are \change{boldface} lower case letters. 
The vector $\q_i$ denotes the $i$-th column of the Matrix $\Q$ whose entries are $q_{ij}$. 
The matrix $\I_n$ is the $n\times n$ identity matrix.
The 2-norm of a vector or the corresponding induced matrix norm is denoted by $\|\cdot\|$.
Symmetric positive definite (positive semidefinite) matrices $\M^\transp = \M$ are denoted by $\M \succ 0$ ($\M \succeq 0$). 
If for two symmetric matrices $\M_2 - \M_1 \succ 0$ ($\succeq 0$), then $\M_1\prec \M_2$ ($\M_1\preceq \M_2$).
The unique square root of a positive semidefinite matrix $\M$ is a positive semidefinite matrix $\N$ such that $\N^2 = \M$. This square root is represented by $\M^{1/2}$.
The time derivative of a function $\x(t)$ is represented by $\dot{\x}(t)$.%

\section{Problem Statement}\label{sec:problemstatement}

The linear time varying system
\begin{subequations}\label{eq:systemED}
	\begin{align}
	\dot \x(t) &= \A(t)\x(t),\qquad\x(t)\in\mathds{R}^n, \label{eq:systemstateED}\\
	\y(t) &= \C(t)\x(t),\qquad\y(t)\in\mathds{R}^p,
	\end{align}
\end{subequations}
is considered for $t \in \mathds J =[0,\infty)$ and the matrices $\A(t)$ and $\C(t)$ of appropriate dimension are assumed to be continuous and uniformly bounded.
System~\eqref{eq:systemstateED} has the unique solution 
$\x(t) = \mtPhi(t,t_0)\x_0,$
where $\mtPhi(t,t_0)$ is the state transition matrix, $\x(t_0)=\x_0\in\mathds{R}^n$ the initial state and $t_0\in\mathds J$ the considered initial time. 
This state transition matrix can be obtained by the associated fundamental matrix differential equation
\begin{equation}\label{eq:fundamental_prel}
\dot{\X}(t) = \A(t) \X(t), \quad \X(t)\in\mathds{R}^{n\times n}.
\end{equation}
Using any solution of~\eqref{eq:fundamental_prel} with $\X(0)=\X_0$ as a non-singular matrix, the state transition matrix is given by
\begin{equation}
\mtPhi(t,t_0) = \X(t)\X^{-1}(t_0).
\end{equation}
For linear systems, stability is entirely characterized by the state transition matrix and the considered stability concept is introduced in the following.
\begin{definition}[uniform exponential stability]\label[definition]{def:UES}
	System~\eqref{eq:systemstateED} is called uniformly exponentially stable if there exists a constant $\mu >0$ and a scalar $K\geq 1$ such that 
	\begin{equation}
	\|\mtPhi(t,t_0)\| \leq K e^{-\mu(t-t_0)}
	\end{equation}
	holds for all $t_0\in\mathds J$ and all $t\geq t_0$.
\end{definition}

The overall goal of the present work is to derive conditions, which guarantee the existence of an observer in the form
\begin{equation}
\dot{\hat{\x}} (t) = \A(t) \hat{\x}(t) + \L(t)\left[\y(t) - \C(t)\hat\x(t)\right]
\end{equation}
such that the resulting estimation error dynamics 
\begin{equation}\label{eq:errorsystem}
\dot \e(t) = \left[\A(t) -\L(t)\C(t)\right]\e(t)
\end{equation}
with $\e(t) = \x(t)-\hat\x(t)$ is uniformly exponentially stable. 
This motivates the definition of uniform exponential detectability introduced in the following, see also~\cite{Ravi1992}.
\begin{definition}[uniform exponential detectability]
	System~\eqref{eq:systemED}, or equivalently the pair $(\A(t),\C(t))$ is called uniformly exponentially detectable if there exists a uniformly bounded output injection gain $\L(t)$ such that system~\eqref{eq:errorsystem} is uniformly exponentially stable. 
\end{definition}
In the following, the term detectability refers to uniform exponential detectability according to this definition.
The ultimate goal is to derive conditions, which guarantee detectability of system~\eqref{eq:systemED} under the assumption, that the system possesses an exponential dichotomy. 
The concept of an exponential dichotomy is introduced in the following.

\section{Exponential Dichotomy and Observability}\label{sec:EDbasics}
This section recalls the basic ideas of exponential dichotomies and observability of linear time varying systems.
\subsection{Exponential Dichotomy}
An exponential dichotomy is a type of conditional exponential stability for time varying linear systems. %
\change{It was introduced by O. Perron and generalizes hyperbolicity in a linear time-invariant setting, i.e., the absence of purely imaginary eigenvalues, to the time-varying case.
More details can be found in ~\cite[Ch. 5]{Kloeden2011},~\cite[Ch. 4, \parsymb 3]{daleckii2002stability} or \cite{Coppel1978,Dieci2007}.}
\change{The concept of an exponential dichotomy} is introduced in the following definition.
\begin{definition}[exponential dichotomy]\label{def:ED}
	System~\eqref{eq:systemstateED} admits an exponential dichotomy (ED) on $\mathds J$ if there exists a projection matrix $\P\in\mathds{R}^{n\times n}$, i.e., a matrix such that $\P^2 =\P$, and constants $K\geq1$ and $\alpha >0$ such that
	\begin{subequations}\label{eq:EDdef}
		\begin{align}
		\|\X(t)\P\X^{-1}(t_0)\| &\leq K e^{-\alpha(t-t_0)} &\text{ for } t\geq t_0\geq 0; \label{eq:EDstable}\\
		\|\X(t)(\I_n-\P)\X^{-1}(t_0)\| & \leq K e^{\alpha (t-t_0)} &\text{ for } 0\leq t\leq t_0,\label{eq:EDunstable}
		\end{align} 
	\end{subequations}
	for some fundamental matrix solution $\X(t)$.
\end{definition}
\change{Under an exponential dichotomy, the set of all solutions $\mathcal S=\{\X(t)\bm\xi:\bm\xi\in \mathds R^n\}$ of system~\eqref{eq:systemstateED} can be written as the direct sum of the sets $\mathcal S_1=\{\X(t)\P\bm\xi:\bm\xi \in \mathds R^n\}$ and \mbox{$\mathcal S_2=\{\X(t)(\I_n-\P)\bm\xi:\bm\xi \in\mathds R^n\}$}.
Moreover, every solution $\x(t)$ in $\mathcal S_1$ satisfies \mbox{$\|\x(t)\|\le K e^{-\alpha(t-t_0)}\|\x(t_0)\|$}, while every solution $\x(t)$ in $\mathcal S_2$ satisfies $\|\x(t)\|\ge K' e^{\alpha(t-t_0)}\|\x(t_0)\|$ for all $t\ge t_0\ge 0$ and some $K'>0$, see~\cite[p. 11]{Coppel1978}.} %
System~\eqref{eq:systemstateED} is uniformly exponentially stable if and only if it admits an exponential dichotomy with $\P=\I_n$, because then~\eqref{eq:EDstable} reduces to a bound on the state transition matrix $\mtPhi(t,t_0)=\X(t)\X^{-1}(t_0)$, which coincides with the definition of uniform exponential stability, see~\Cref{def:UES}.

It should be remarked that any projection matrix of rank $k\leq n$ is similar to the projection matrix
\begin{equation}\label{eq:Pidentity}
\P = \begin{bmatrix} \vec 0_{(n-k)\times(n-k)} & \vec 0 \\ \vec 0 & \I_k\end{bmatrix},
\end{equation}
see~\cite{Coppel1978}.
Hence, without loss of generality, it can be assumed that $\P$ is already in the form of~\eqref{eq:Pidentity}.
To see this, let~\eqref{eq:systemED} have an exponential dichotomy with some fundamental matrix $\X(t)$ and a projection matrix $\tilde\P$ with $\rank \tilde\P=k$. 
For a suitable similarity transformation matrix $\U$ it follows that $\P = \U^{-1} \tilde\P \U$, where $\P$ is in the form~\eqref{eq:Pidentity}.
By rewriting~\eqref{eq:EDstable}, one directly obtains
\begin{equation}
\begin{aligned}
&\|\X(t)\U\P\U^{-1}\X^{-1}(t_0)\| = \\&\|\Y(t)\P\Y^{-1}(t_0) \|  \leq K e^{-\alpha(t-t_0)} \text{  for } t\geq t_0\geq 0;
\end{aligned}
\end{equation}
with the new fundamental solution $\Y(t) = \X(t)\U$. 
This analogously holds for~\eqref{eq:EDunstable}.
Hence, ~\eqref{eq:systemED} also has an exponential dichotomy with a projection matrix $\P$ in the form of~\eqref{eq:Pidentity}.

A state transformation, which preserves the stability properties of the system is called a Lyapunov transformation~\cite[\change{Definition 3.1.1}]{Adrianova1995}.
It is a smooth and invertible linear change of coordinates $\x(t)=\T(t)\z(t)$ with $\T(t)$, $\T^{-1}(t)$ and $\dot{\T}(t)$ uniformly bounded for all $t\in\mathds J$.
The transformed system is given by
\begin{equation}\label{eq:Lyaptransfsystem}
\dot{ \z}(t) = \left[\T^{-1}(t)\A(t)\T(t) \change{-} \T^{-1}(t)\dot{\T}(t)\right] \z(t).
\end{equation}
Systems~\eqref{eq:systemstateED} and~\eqref{eq:Lyaptransfsystem} are also called kinematically similar.
\subsection{Observability}
An important concept for the subsequent detectability analysis is uniform complete observability, which was introduced by Kalman~\cite{Kalman1961}.

\begin{definition}[uniform complete ob\-serv\-ability]\label{def:UCC}
		The symmetric positive semidefinite $n\times n$ matrix 
	\begin{equation}\label{eq:obsv_gramian}
	\M(t_1,t_0) = \integ{s}{t_0}{t_1}{\mtPhi^\transp(s,t_0)\C^\transp(s)\C(s)\mtPhi(s,t_0)}
	\end{equation}
	is the so-called observability Gramian.
	System~\eqref{eq:systemED} is called uniformly completely observable, if there exist positive constants $\beta_1$, $\beta_2$ and $\sigma$ such that
	\begin{equation}\label{eq:UCO}
	\beta_1 \I_n \preceq \M(t_0+\sigma, t_0) \preceq \beta_2 \I_n \text{ for all }t_0\in\mathds J.
	\end{equation} 	

\end{definition}
\change{The upper bound $\beta_2$ always exists due to the boundedness assumptions on the coefficient matrices.}
Uniform complete observability of~\eqref{eq:systemED} is sufficient for detectability, see, e.g.,~\cite{Tranninger2020}.
A result presented in~\cite[Lemma 1]{Zhang2015} states that uniform complete observability is preserved under output injection.
\begin{lemma}[uniform complete observability under output injection,~\cite{Zhang2015}] \label{thm:UCOpreservation}
	The pair $(\A(t),\C(t))$ is uniformly completely observable if and only if for any bounded and integrable matrix $\L(t)$, the pair \mbox{$(\A(t)-\L(t)\C(t),\C(t))$} is uniformly completely observable.
\end{lemma}

According to~\Cref{thm:UCOpreservation} and \cite[Lemma 2]{Zhang2015}, this also guarantees the existences of constants $\bar \beta_1,\, \bar \beta_2>0$ such that
\begin{equation}
\bar \beta_1\I_n \preceq \integ{s}{t_0}{t_0+\sigma}{ \mtPhi_e^\transp(s,t_0)\C^\transp(s) \C(s) \mtPhi_e(s,t_0)} \preceq \bar \beta_2\I_n
\end{equation}
holds for the same $\sigma$ and for all $t_0\in\mathds J$, where the matrix $\mtPhi_e(\cdot,\cdot)$ is the state transition matrix of~\eqref{eq:errorsystem}.

\section{Uniform Exponential Detectability of Systems with Exponential Dichotomy}\label{sec:Detectability}
In the following, mainly systems in triangular or block triangular form are considered. 
This is motivated by the fact that \change{uniform exponential detectability is invariant with respect to Lyapunov transformations.
Moreover,} every linear system is similar to an upper triangular system by an orthogonal change of coordinates introduced in the following together with the main result.

\subsection{Necessary and Sufficient Condition for Uniform Exponential Detectability}
Consider $\X(t)$ as the fundamental solution in~\eqref{eq:EDdef}.
Based on Perron's lemma~\change{\cite[Theorem 3.3.1 \& Remark 3.3.2]{Adrianova1995}}, there exists an orthogonal Lyapunov transformation \mbox{$\R(t)=\Q^\transp(t) \X(t)$} such that $\R(t)$ is upper triangular with a positive diagonal. 
The transformation to upper triangular form can be obtained by means of the continuous QR decomposition~\change{\cite{Dieci2007}}.
The matrices $\Q$ and $\R$ are the solutions of the differential equations
\begin{align}
\dot{\R}(t) &=  \B(t)\R(t),   &\B &= \Q^\transp \A \Q-\S, \label{eq:fullqr1}\\
\dot{\Q}(t) &= \Q(t)\S(t),   &\S&=-\S^\transp,\; \label{eq:fullqr2}
\end{align}
with skew-symmetric matrix $\S$ according to \mbox{$s_{ij} = \q_i^\transp \A \q_j$},  $i>j$, $\Q = [\q_1, \ldots, \q_n]$ and a bounded upper triangular matrix $\B$. 
The initial  condition is $\X(0) = \Q(0)\R(0)$, where $\Q(0)$ and $\R(0)$ can be obtained by the QR-decomposition. 
Note that $\Q$ and $\R$ are uniquely defined if the diagonal of $\R$ is positive~\cite{Dieci2007}.

The transformation to upper triangular form is especially useful for determining if a given system possesses an exponential dichotomy. 
Algorithms for the computation of the projection matrix $\P$ are proposed in~\cite{Dieci2010a,Dieci2011}.

If system~\eqref{eq:systemED} possesses an exponential dichotomy with a fundamental solution $\X(t)$ and a projection matrix $\P$ in the form~\eqref{eq:Pidentity}, then also the transformed system $\dot \z(t) = \B(t)\z(t)$ with $\z(t) = \Q^\transp(t)\x(t)$ has an exponential dichotomy with the same projection matrix $\P$ and the fundamental matrix solution $\R(t) = \Q^\transp(t) \X(t)$.
This can be seen by using $\X(t) = \Q(t)\R(t)$ in the definition of an exponential dichotomy~\eqref{eq:EDdef}. 
Considering~\eqref{eq:EDstable}, this yields
\begin{equation}\label{eq:EDtransf_presv}
\begin{aligned}
&\|\Q(t)\R(t)\P\R^{-1}(t_0)\Q^\transp(t_0)\| = \\
& \|\R(t)\P\R^{-1}(t_0)\| \leq K e^{-\alpha(t-t_0)} \text{ for } t\geq t_0\geq 0,
\end{aligned}
\end{equation}
which follows directly from the orthogonality of $\Q$. 
An analogous statement holds for~\eqref{eq:EDunstable}.
This allows to state the main result directly for the class of block upper triangular systems.

\begin{theorem}[detectability of block triangular systems]\label{thm:mainresult}
	Let a system $\dot \x(t) = \B(t)\x(t),\;\x(t)\in\mathds R^n$ have a block triangular structure partitioned according to 
	\begin{subequations}\label{eq:blocktriuED}
		\begin{equation}\label{eq:EDblocktriu2}
		\begin{bmatrix}\dot \x_1(t)\\ \dot \x_2(t) \end{bmatrix} = \begin{bmatrix}\B_{11}(t) & \B_{12}(t) \\ \vec 0 & \B_{22}(t) \end{bmatrix} \begin{bmatrix}\x_1(t) \\ \x_2(t) \end{bmatrix}
		\end{equation}
		with $\x_1(t)\in\mathds R^{n-k}$ and $\x_2(t)\in\mathds R^k$.
		The system output is given by
		\begin{equation}
		\y(t) = \begin{bmatrix}\C_1(t) & \C_2(t) \end{bmatrix}\begin{bmatrix}\x_1(t) \\ \x_2(t) \end{bmatrix}.
		\end{equation}
	\end{subequations}
	The block matrices $\B_{11}(t)$, $\B_{12}(t)$, $\B_{22}(t)$, $\C_1(t)$ and $\C_2(t)$ are uniformly bounded matrices of appropriate dimensions. 
	It is assumed, that~\eqref{eq:EDblocktriu2} possesses an exponential dichotomy with $\P$ in the form of~\eqref{eq:Pidentity}.
	
	Then, system~\eqref{eq:blocktriuED} is uniformly exponentially detectable if and only if the pair $(\B_{11}(t),\C_1(t))$ is uniformly completely observable.
\end{theorem}

In order to \change{prove} this result, additional tools for the analysis of systems with exponential dichotomy are presented in the following.

\subsection{Anti-Stable and Block Diagonal Systems}

An important property of systems with an exponential dichotomy is the so-called reducibility to block diagonal form~\change{\cite[Chapter 5]{Coppel1978}}. 
Unlike the transformation to triangular form, this transformation is possibly badly conditioned numerically.
Nevertheless, this property will be used in the following sections to obtain insight into the introduced detectability condition. 
\begin{definition}[reducibility]
	System~\eqref{eq:systemstateED} is reducible to block diagonal form with dimension $k$, if there exists a Lyapunov transformation $\x(t)=\S(t)\z(t)$, which transforms~\eqref{eq:systemstateED} to the block diagonal system
	\begin{equation}\label{eq:EDblockdiagonal}
	\dot \z(t) = \D(t) \z(t),  \qquad 	\D(t) = \begin{bmatrix} \D_1(t) &\vec 0\\ \vec 0  & \D_2(t)\end{bmatrix},
	\end{equation}
	where $\D_2(t)$ is a $k\times k$ matrix.
\end{definition}
A result presented in~\cite{Coppel1967}, which guarantees reducibility to block diagonal form for systems with exponential dichotomy, is summarized in the following.
\begin{lemma}[reducibility for systems with exponential dichotomy]
	Let system~\eqref{eq:systemstateED} have an exponential dichotomy with $\P$ in the form of~\eqref{eq:Pidentity} and a corresponding fundamental matrix solution $\X(t)$. 
	Then, there exists a Lyapunov transformation 
	\begin{equation}\label{eq:EDtransf1}
	\S(t) = \X(t)\T^{-1}(t)
	\end{equation}
	with a symmetric positive definite $\T(t)$ such that
	\begin{equation}\label{eq:EDtransf2}
	\T^2 (t)= \P\X^\transp(t)\X(t)\P + (\I-\P)\X^\transp(t)\X(t)(\I-\P).
	\end{equation}
	This transformation reduces~\eqref{eq:systemstateED} to the block diagonal system~\eqref{eq:EDblockdiagonal} with $\D_2(t)$ as a matrix of dimension $k\times k$. 
	Moreover, $\dot \z(t)=\D(t)\z(t)$ has an exponential dichotomy with the same projection matrix $\P$ \change{and the same $\alpha$}.
\end{lemma}
\change{\begin{proof}
        Since $\| \X(t)\P\X^{-1}(t) \| \leq K$, \cite[Lemma 1]{Coppel1967} and its proof guarantee that $\S(t)$ and its inverse are bounded.
        Furthermore,~\cite[Lemma 2]{Coppel1967} and its proof show that boundedness of $\A(t)$ implies that $\dot \S(t) = \A(t) \S(t) - \S(t) \D(t)$ is bounded, and that $\D(t)$ commutes with $\P$ in \eqref{eq:Pidentity}.
        Hence, \eqref{eq:EDtransf1} is a Lyapunov transformation and $\D(t)$ is block diagonal.
		The last statement then follows from~\cite[Lemma 3]{Coppel1967}.
\end{proof} }
For the two special cases $\P=\vec 0$ and $\P=\I_n$ one gets the straightforward relation $\D(t) = \A(t)$ and hence these two trivial cases are neglected in the following.
The system in block diagonal form is decoupled and the two independent systems
\begin{subequations}
	\begin{align}
	\dot \z_1(t) &= \D_1(t) \z_1(t)\label{eq:EDdiag1}\\
	\dot \z_2(t) &= \D_2(t) \z_2(t) \label{eq:EDdiag2}
	\end{align}
\end{subequations}
are of order $n-k$ and $k$, respectively. 
System~\eqref{eq:EDdiag1} is a so-called anti-stable system, because it has an exponential dichotomy with projection $\P_1=\vec 0$ and all trajectories grow uniformly and exponentially.
The second system~\eqref{eq:EDdiag2} has an exponential dichotomy with $\P_2 = \I_k$ and hence it is uniformly exponentially stable.
The following result presents a detectability condition for anti-stable systems.
\begin{proposition}[uniform exponential detectability of anti-stable systems]\label{thm:UCOisdetectability}
	Let~\eqref{eq:systemstateED} admit an exponential dichotomy with $\P=\vec 0$. 
	Then, system~\eqref{eq:systemED} is uniformly exponentially detectable if and only if it is uniformly completely observable.
\end{proposition}
\begin{proof}
	Sufficiency follows, e.g., from the observer design presented in~\cite{Tranninger2020,Kalman1961}, which guarantees the existence of an observer.
	For necessity, let the system be uniformly exponentially detectable but not uniformly completely observable. 
	With the aid of the variational equation, the solution of~\eqref{eq:errorsystem} can be stated as
	\begin{equation}
	\e(t) = \mtPhi(t,t_0)\e(t_0) - \integ{s}{t_0}{t}{\mtPhi(t,s)\L(s)\C(s)\mtPhi_e(s,t_0)\e(t_0)}.
	\end{equation}
	Hence, the relation
	\begin{equation}\label{eq:UCOproof1}
	\mtPhi_e(t_2,t_1) = \mtPhi(t_2,t_1) - \integ{s}{t_1}{t_2}{\mtPhi(t_2,s)\L(s)\C(s)\mtPhi_e(s,t_1)}
	\end{equation}
	holds for all $t_1$, $t_2$. 
	Multiplication of~\eqref{eq:UCOproof1} with $\mtPhi(t_1,t_2)$ gives
	\begin{equation}\label{eq:UCOproof1f}
	\mtPhi(t_1,t_2)\mtPhi_e(t_2,t_1) = \I - \integ{s}{t_1}{t_2}{\mtPhi(t_1,s)\L(s)\C(s)\mtPhi_e(s,t_1)}.
	\end{equation}
	For \change{any} $s\geq t_1$, one can bound $\mtPhi(t_1,t_2)$ according to
	\begin{equation}
	\|\mtPhi(t_1,s)\|\leq Ke^{-\alpha(s-t_1)} \leq K
	\end{equation}
	for some $K,\;\alpha>0$, see~\eqref{eq:EDunstable}.
	By the uniform exponential detectability assumption, system~\eqref{eq:errorsystem} is uniformly exponentially stable for a bounded $\|\L(t)\|\leq K_L$ and
	\begin{equation}\label{eq:bound_error}
	\|\mtPhi_e(t_2,t_1)\| \leq K_e e^{-\mu(t_2-t_1)}
	\end{equation}
	holds for all $t_2\geq t_1$ and some $\mu>0$ and $K_e\geq 1$. 
	Together with the bounds on $\mtPhi$, $\mtPhi_e$ and $\L(t)$, the multiplication of~\eqref{eq:UCOproof1f} with some vector $\boldsymbol{\xi}$ of appropriate dimension and taking the norm on both sides results in the inequality 
	\begin{equation}\label{eq:UCOproof2}
	KK_e\|\boldsymbol{\xi}\| e^{-\mu(t_2-t_1)} \geq \|\boldsymbol{\xi}\| - KK_L \integ{s}{t_1}{t_2}{\|\C(s)\mtPhi_e(s,t_1)\boldsymbol{\xi}\|}.
	\end{equation}
	
	By assumption, the pair $(\A(t),\C(t))$ is not uniformly completely observable and hence according to~\Cref{thm:UCOpreservation}, the pair $(\A(t)-\L(t)\C(t),\C(t))$ is also not uniformly completely observable.
	By negating the statements in~\Cref{def:UCC}, it follows that for any $\sigma>0$ and any $\change{\beta_1}>0$ there exists a non-trivial vector $\boldsymbol{\eta}\change{\in\mathds R^n}$ \change{and a $t_0\in\mathds J$} such that the inequality
	\begin{equation}\label{eq:EDnotUCO}
	\boldsymbol{\eta}^\transp  \integ{s}{t_0}{t_0+\sigma}{ \mtPhi_e^\transp(s,t_0)\C^\transp(s) \C(s) \mtPhi_e(s,t_0)} \boldsymbol{\eta} < \change{\beta_1} \|\boldsymbol{\eta}\|^2
	\end{equation}
	or equivalently
	\begin{equation}\label{eq:2222}
	\integ{s}{t_0}{t_0+\sigma}{\| \C(s) \mtPhi_e(s,t_0)\boldsymbol{\eta}\|^2} < \change{\beta_1} \|\boldsymbol{\eta}\|^2
	\end{equation}
	is fulfilled.
	By applying Schwarz's inequality, 
	one obtains
	\begin{equation}\label{eq:UCOproof3}
	\integ{s}{t_0}{t_0+\sigma}{\| \C(s) \mtPhi_e(s,t_0)\bm\eta \|} < \sqrt{\change{\beta_1}\sigma}\|\bm \eta\|.
	\end{equation}
	Now, let $t_1=t_0$ and $t_2=t_0+\sigma$ in~\eqref{eq:UCOproof2} and choose $\boldsymbol\xi=\boldsymbol{\eta}$ and $t_0$ such that~\eqref{eq:EDnotUCO} is fulfilled.
	By selecting $\sigma = \frac{\ln(3KK_e)}{\mu}$ and $\change{\beta_1}=(9K^2K_L^2\sigma)^{-1}$ and combining~\eqref{eq:UCOproof3} with~\eqref{eq:UCOproof2}
	one obtains
	\begin{equation}
	\frac{1}{3}\|\boldsymbol{\eta}\| = KK_e\|\boldsymbol{\eta}\| e^{-\mu\sigma} \geq \|\boldsymbol{\eta}\| - KK_L\sqrt{\change{\beta_1}\sigma}\|\boldsymbol{\eta}\| = \frac{2}{3}\|\boldsymbol{\eta}\|.
	\end{equation}
	This is a contradiction and hence uniform exponential detectability implies uniform complete observability for anti-stable systems.
\end{proof}
The proof was inspired by the proof of Theorem 3 in~\cite{Ikeda1972}. 
There, the goal was to show that for systems with bounded coefficient matrices, the dual concept uniform complete controllability is equivalent to complete stabilizability with arbitrary decay rate. 
A key difference to the present proof is that the decay rate $\mu$ in~\cite{Ikeda1972} has to be chosen in an appropriate way.
This is avoided in the proof of~\Cref{thm:UCOisdetectability} by utilizing the fact that the system is anti-stable. 

To sum up, \Cref{thm:UCOisdetectability} states that uniform complete observability is the minimum requirement in order to obtain a uniformly exponentially stable observer error system by a bounded feedback gain for anti-stable systems. 
This can be extended to general systems with exponential dichotomy.
It can be assumed without loss of generality, that the system is already in block diagonal form
\begin{subequations}\label{eq:blockdiag_detectability}
\begin{align}\label{eq:sys_blockdiag_detectability}
\begin{bmatrix}
\dot \z_1(t) \\
\dot \z_2(t) 
\end{bmatrix} &= \begin{bmatrix}
\D_1(t) & \vec 0 \\
\vec 0 & \D_2(t) 
\end{bmatrix} \begin{bmatrix}
\z_1(t) \\ \z_2(t) 
\end{bmatrix} \\
\y(t) &= \begin{bmatrix}
\C_{1}(t) & \C_{2}(t)
\end{bmatrix}\begin{bmatrix}
\z_1(t) \\ \z_2(t) 
\end{bmatrix}
\end{align}
\end{subequations}
with $\z_2(t)\in\mathds{R}^k$. 
It is assumed that~\eqref{eq:sys_blockdiag_detectability} has an exponential dichotomy with the projection matrix $\P$ in the form of~\eqref{eq:Pidentity} and matrices $\D_1(t)$, $\D_2(t)$, $\C_{1}(t)$ and $\C_{2}(t)$ of appropriate dimension.
The following result follows from~\Cref{thm:UCOisdetectability}.
\begin{proposition}\label{thm:DetectabilityDiagonal}
	System~\eqref{eq:blockdiag_detectability} is uniformly exponentially detectable, if and only if the pair $(\D_1(t),\C_{1}(t))$ is uniformly completely observable.
\end{proposition}
\begin{proof}
	For necessity, assume that the pair $(\D_1(t),\C_1(t))$ is not uniformly completely observable.
	Let the initial condition for $\z_2$ be $\z_2(t_0)=\vec 0$. 
	Hence, system~\eqref{eq:blockdiag_detectability} reduces to 
	$\dot{\z}_1(t) = \D_1(t)\z_1(t)$,  
	$\y(t) = \C_1(t)\z_1(t)$  
for $t\geq t_0$. 
This anti-stable system is not uniformly completely observable and hence not uniformly exponentially detectable, see~\Cref{thm:UCOisdetectability}.
For sufficiency, \change{assume that $(\D_1(t),\C_1(t))$ is uniformly completely observable and consider an observer of the form}
	\begin{subequations}\label{eq:obsv_sufficient}
		\begin{align}
			\dot{\hat{\z}}_1(t) &= \D_1(t)\hat{\z}_1(t) + \L_1(t)\left[\y(t)-\change{\C(t)\hat\z(t)}\right], \\
			\dot{\hat{\z}}_2(t) &= \D_2(t)\hat\z_2(t).
		\end{align}
	\end{subequations}
		\change{The dynamics of the estimation error $\e(t)=\z(t)-\hat\z(t)$ can be derived as}
		\begin{subequations}\label{eq:errorsystem_partitioned}
		\begin{align}
		\dot \e_1(t) &= \left[\D_1(t)-\L_1(t)\C_1(t)\right]\e_1(t) \change{-} \L_1(t)\C_2(t)\change{\e_2(t)}\label{eq:e1} \\
		\dot \e_2(t) &=  \D_2(t)\change{\e_2(t)}.\label{eq:e2}
		\end{align}
		\end{subequations}
	Equation~\eqref{eq:e2} is uniformly exponentially stable. 
	According to~\Cref{thm:UCOisdetectability}, there exists a uniformly bounded feedback gain $\L_1(t)$ such that the unperturbed error system \mbox{$\dot \e_1(t) = \left[\D_1(t)-\L_1(t)\C_1(t)\right]\e_1(t)$} is uniformly exponentially stable. 
	\change{
	Hence, for the uniformly bounded feedback gain $\L(t) = [\L_1^\transp(t) \; \vc 0^\transp]^\transp$, the block triangular error system $\dot\e(t) = [\D(t)-\L(t)\C(t)]\e(t)$, i.e.,~\eqref{eq:errorsystem_partitioned}, is uniformly exponentially stable, see~\cite[Theorem 2]{Zhou2016}.
	Thus, system~\eqref{eq:blockdiag_detectability} is uniformly exponentially detectable.
	}
\end{proof}
The transformation~\eqref{eq:EDtransf1}, which brings system~\eqref{eq:systemED} to block diagonal form is hard to obtain in practice, because it requires the \change{unbounded} solution of the fundamental matrix differential equation.
The transformation to triangular form requires only to solve the orthogonal differential equation~\eqref{eq:fullqr2}. 
As already shown, the projection matrix $\P$ is preserved under the transformation to triangular form. 
A relation of block triangular systems to the corresponding reduced block diagonal form is presented in the following.

\begin{proposition}[exponential dichotomy of block triangular systems]\label{thm:EDblocktriu}
	Let a system $\dot \x(t) = \B(t)\x(t),\;\x(t)\in\mathds R^n$ have a block triangular structure partitioned according to 
	\begin{equation}\label{eq:EDblocktriu}
	\begin{bmatrix}\dot \x_1(t)\\ \dot \x_2(t) \end{bmatrix} = \begin{bmatrix}\B_{11}(t) & \B_{12}(t) \\ \vec 0 & \B_{22}(t) \end{bmatrix} \begin{bmatrix}\x_1(t) \\ \x_2(t) \end{bmatrix}.
	\end{equation}
	The block matrices $\B_{11}(t)$, $\B_{12}(t)$ and $\B_{22}(t)$ are of dimension $(n-k)\times(n-k)$, $(n-k)\times k$ and $k\times k$, respectively, with $1\leq k\leq n-1$.
	It is assumed that $\dot \x_1(t) = \B_{11}(t)\x_1(t)$ has an exponential dichotomy with $\P_1 = \vec 0$ and $\dot \x_2(t)=\B_{22} \x_2(t)$ has an exponential dichotomy with $\P_2 = \I_k$.
	Then, system~\eqref{eq:EDblocktriu} has an exponential dichotomy with the projection 
	\begin{equation}\label{eq:Pblock}
	\P = \begin{bmatrix} \vec 0 & \vec 0 \\ \vec 0 & \I_k\end{bmatrix}.
	\end{equation}
	Moreover, it is reducible to the block diagonal form~\eqref{eq:EDblockdiagonal} with $\D_1(t) = \B_{11}(t)$.
\end{proposition}
\begin{proof}
	According to~\cite[Theorem 24]{Dieci2010a}, system~\eqref{eq:EDblocktriu} has an exponential dichotomy with projection
	\begin{equation}\label{eq:Pblockprf}
	\P = \begin{bmatrix} \vec 0 & \vec 0 \\ \vec 0 & \I_k\end{bmatrix}
	\end{equation}
	and the fundamental matrix solution
	\begin{equation}\label{eq:fundtriu}
	\X(t) = \begin{bmatrix}  \X_{11}(t) & \X_{12}(t) \\ \vec 0 & \X_{22}(t) \end{bmatrix}.
	\end{equation}
	The matrices $\X_{11}(t)$ and $\X_{22}(t)$ are any non-singular fundamental matrix solutions of the systems \mbox{$\dot\x_1(t) =\B_{11}(t)\x_1(t)$} and $\dot\x_2(t)=\B_{22}(t)\x_2(t)$, respectively.
	The matrix $\X_{12}(t)$ is given by
	\begin{equation}
	\X_{12}(t) = -\X_{11}(t)\integ{\tau}{t}{\infty}{\X_{11}^{-1}(\tau)\B_{12}(\tau)\X_{22}(\tau)}.
	\end{equation}
	Hence, it remains to be shown that the transformation to the block diagonal form~\eqref{eq:EDblockdiagonal} does not change the upper block, i.e., \change{that} $\D_1(t) = \B_{11}(t)$. 
	The transformation matrix is given by~\eqref{eq:EDtransf1} with fundamental matrix solution~\eqref{eq:fundtriu}.
	It follows from~\eqref{eq:EDtransf2} with $\P$ as in~\eqref{eq:Pblockprf} that
	\begin{equation}
	\T^2(t)= \begin{bmatrix} \X_{11}^\transp(t)\X_{11}(t) &\vec 0 \\ \vec 0 & \N^\transp(t)\N(t) \end{bmatrix}
	\end{equation}
	with $\N^\transp(t)\N(t) = \X_{12}^\transp(t)\X_{12}(t) + \X_{22}^\transp(t)\X_{22}(t)$. 
	Hence, $\T(t)$ is block diagonal with $\T(t) = \text{diag}\,(\X_{11}(t),\N(t))$
	and the transformation matrix is given by
	\begin{equation}\label{eq:diag1}
\begin{aligned}
	\S(t) &= \X(t)\T^{-1}(t) 
	= \begin{bmatrix} \I_{n-k} & \X_{12}(t)\N^{-1}(t) \\ \vec 0 & \X_{22}(t)\N^{-1}(t)   \end{bmatrix}.
\end{aligned}
	\end{equation}
	Its inverse is given by
	\begin{equation}\label{eq:diag2}
	\renewcommand*{\arraystretch}{1.2}
	\S^{-1}(t) = \begin{bmatrix}\I_{n-k} & -\X_{12}(t)\X^{-1}_{22}(t) \\ \vec 0 & \N \X_{22}^{-1}(t) \end{bmatrix}
	\end{equation}
	and the time derivative of $\S(t)$ can be stated as
	\begin{equation}\label{eq:diag3}
	\renewcommand*{\arraystretch}{1.2}
	\dot \S(t) = \begin{bmatrix}\vec 0 & \dt{\phantom .}\left[\X_{12}(t)\N^{-1}(t)\right] \\ \vspace{0.1cm} \vec 0 & \dt{\phantom .}\left[\X_{22}(t)\N^{-1}(t)\right]\end{bmatrix}
	\end{equation}
	By a straightforward computation using~\eqref{eq:diag1},~\eqref{eq:diag2} and~\eqref{eq:diag3} in
$	\D(t) = \S^{-1}(t)\B(t)\S(t) - \S^{-1}(t)\dot\S(t),$
	one can see that the upper left block remains unchanged for an upper block triangular coefficient matrix and hence $\D_1(t)=\B_{11}(t)$.
\end{proof}
The previous results allow to state the proof of the main result.
\subsection{Proof of {{\Cref{thm:mainresult}}}}
	It follows from~\Cref{thm:EDblocktriu} that system~\eqref{eq:EDblocktriu2} is reducible to block diagonal form $\dot \z(t)=\D(t)\z(t)$ with $\D_1(t)=\B_{11}(t)$.
	The block diagonal system has an exponential dichotomy with the same projection matrix $\P$.
	It is shown in the proof of~\Cref{thm:EDblocktriu} that the transformation matrix is given by
	\begin{equation}
	\S(t) = \begin{bmatrix} \I_{n-k} & \S_{12}(t) \\ \vec 0 & \S_{22}(t)\end{bmatrix}.
	\end{equation}
	The matrices $\S_{12}(t)$ and $\S_{22}(t)$ are stated in~\eqref{eq:diag1}. 
	The important point in this proof is that 
	\begin{equation}
	\begin{aligned}
	\y(t) &=  \C(t)\S(t)\z(t) \\
	&= \begin{bmatrix}\C_1(t) & \C_2(t) \end{bmatrix} \begin{bmatrix}\I_{n-k} & \S_{12}(t) \\ \vec 0 & \S_{22}(t) \end{bmatrix}\begin{bmatrix}\z_1(t) \\\z_2(t) \end{bmatrix}\\
	&=\begin{bmatrix} \C_1(t) & \tilde\C_2(t)\end{bmatrix}\begin{bmatrix}\z_1(t) \\\z_2(t) \end{bmatrix},
	\end{aligned} 
	\end{equation}
	with $\tilde\C_2(t)=\C_1(t)\S_{12}(t)+\C_2(t)\S_{22}(t)$.
	The rest of the proof follows from the proof of~\Cref{thm:DetectabilityDiagonal}. \hfill\QED
\balance

\section{Conclusion and Outlook}\label{sec:conclusion}
This work presents necessary and sufficient conditions for uniform exponential detectability of systems, which admit an exponential dichotomy. 
In particular, the cases of diagonal, anti-stable, and upper block triangular systems are considered in detail.
The latter form is of particular importance, because it can be obtained in a numerically well-conditioned way by means of a continuous QR decomposition.

Future research aims at the combination of the presented detectability conditions with numerical tools for the stability analysis. 
\change{Moreover, it would be interesting to combine results on the robustness of uniform complete observability~\cite{Sastry1982} with the roughness property of exponential dichotomies to investigate similar properties in a robust framework.
This may also allow the extension of the ideas to certain classes of nonlinear systems.
}

\balance

\bibliographystyle{IEEEtran}
\bibliography{IEEEabrv,references}

\end{document}